\documentclass[10pt,a4paper]{article}
\usepackage[T1]{fontenc}
\usepackage[utf8]{inputenc}
\usepackage{authblk}

\title{Gain-Loss-Driven travelling waves in PT-Symmetric Nonlinear Metamaterials}
\author[1,4]{M. Agaoglou\thanks{makrina{\_}agao@hotmail.com}}
\author[2,3]{M.\ Fe\v{c}kan\thanks{michal.feckan@gmail.com}}
\author[2,3]{M.\ Posp\'{i}\v{s}il\thanks{michal.pospisil@mat.savba.sk}}
\author[4]{V.M. Rothos\thanks{rothos@auth.gr}}
\author[5]{H.\ Susanto\thanks{hsusanto@essex.ac.uk}}

\affil[1]{Department of Mathematics and Statistics\\ University of Massachusetts, Lederle Graduate Research Tower, Amherst, MA 01003, USA}

\affil[2]{Department of Mathematical Analysis and Numerical Mathematics, Comenius University in Bratislava, Mlynsk\'a dolina, 842 48 Bratislava, Slovakia}

\affil[3]{Mathematical Institute of Slovak Academy of Sciences, \v{S}tef\'anikova 49, 814 73 Bratislava, Slovakia}

\affil[4]{Lab of Nonlinear Mathematics and Department of Mechanical Engineering, Faculty of Engineering,\\ Aristotle University of Thessaloniki, Thessaloniki 54124, Greece}

\affil[5]{Department of Mathematical Sciences, University of Essex, Wivenhoe Park, Colchester CO4 3SQ, United Kingdom}

\usepackage{graphicx}
\usepackage{epstopdf,subfigure}
\usepackage{amsmath}
\usepackage{amssymb}
\usepackage{amsthm}

\newtheorem{theorem}{Theorem}[section]
\newtheorem{lemma}[theorem]{Lemma}

\numberwithin{equation}{section}

\DeclareMathOperator{\sech}{sech}
\DeclareMathOperator{\csch}{csch}

\def\R{\mathbb R}

\def\eu{\mathrm{e}}

\usepackage{amssymb}

\usepackage[figuresright]{rotating}

\begin{document}

\maketitle
%\title{Gain-Loss-Driven travelling waves in PT-Symmetric Nonlinear Metamaterials}

\begin{abstract}
In this work we investigate a one-dimensional parity-time (PT)-symmetric magnetic metamaterial consisting of split-ring dimers having gain or loss. Employing a Melnikov analysis we study the existence of localized travelling waves, i.e.\ homoclinic or heteroclinic solutions. We find conditions under which the homoclinic or heteroclinic orbits persist. Our analytical results are found to be in good agreement with direct numerical computations. For the particular nonlinearity admitting travelling kinks, numerically we observe homoclinic snaking in the bifurcation diagram. The Melnikov analysis yields a good approximation to one of the boundaries of the snaking profile.
\end{abstract}




%\end{frontmatter}

% \linenumbers

\section{Introduction}
Metamaterials are composite materials that are engineered structurally (rather than chemically) to have exotic properties that may not be found in nature. They are artificial materials made of ``artificial atoms" (called meta-atoms) that are placed periodically. Together with the properties of the ``atoms", the geometry (e.g., the perfectly periodic arrangement) then creates non-natural, but desirable effective behaviours. The development of this new type of materials was due to the work of Pendry \cite{Pendry1,Pendry2}, following the proposal of Veselago 
\cite{Veselago} on negative refraction. While Veselago studied continuous systems in his pioneering article, Pendry, on the other hand, started from discrete systems that lead to metamaterials \cite{Smith}. Because Pendry considered a mean field (i.e., homogenization) approach that lead him to the findings, certain conditions on the wavelength and the characteristic length of metamaterials should hold
in order to validate the description.

Another paradigm of artificial materials is the parity-time (PT) 
symmetric systems that belong to a class of synthetic materials that do 
not obey separately the parity (P) and time (T) symmetries but instead 
they do exhibit a combined PT-symmetry. The ideas and notions of PT-
symmetric systems have originated from the extension of ordinary 
Quantum Mechanics to PT-symmetric (i.e., non-Hermitian) Hamiltonians 
that have been studied for many years. Recently it has been realized 
that many classical systems are PT-symmetric \cite{Bender}. 
Subsequently, the notion of PT-symmetry has been extended to dynamical 
lattices, particularly in optics \cite{El-Ganainy,Makris}. It was 
immediately after that the theory evolved to include nonlinear lattices 
\cite{Dmitriev}. %Metamaterials in many cases suffer. However, 
Building metamaterials with PT-symmetry, relying on a delicate balance between gain and loss where loss may be matched with an equal amount of gain, may provide a way out from high losses typically present in many cases and result in altogether new functionalities and electromagnetic characteristics \cite{laza13}.

Consider a one-dimensional array of dimers, each comprising two nonlinear split-ring resonators (SRRs): one with loss and the other with an equal amount of gain as sketched in Fig.\ \ref{lazaridisdimernew}. The SRRs are coupled magnetically and/or electrically through dipole-dipole forces \cite{Sydoruk,Hesmer,Sersic,Rosanov} and are regarded as RLC circuits, featuring a resistance $R$, an inductance $L$, and a capacitance $C$. A tunnel (Esaki) diode forms a nonlinear metamaterial element with gain. The diodes exhibit a well-defined negative resistance region in their current-voltage characteristics that has a characteristic ``N'' shape. A bias voltage applied to the diode can move its operation point in the negative resistance region and then the SRR-diode system gains energy from the source. The presence of these elements, in addition to providing gain also introduces nonlinearity in the metamaterial.

The alternating magnetic field induces an electromotive force in each SRR due to Faraday's law which in turn produces currents that couple the SRRs magnetically through their mutual inductance. The coupling strength between SRRs is rather weak due to the nature of their interaction (magnetoinductive), and has been calculated accurately \cite{Sydoruk,Rosanov}. The SRRs may also be coupled electrically through the electric dipoles that develop in their slits. Thus, in the general case one has to consider both magnetic and electric coupling between SRRs. However, for particular relative orientations of the SRR slits the magnetic interaction is dominant, while the electric interaction can be neglected in a first approximation \cite{Hesmer,Sersic,Feth}. %The PT-symmetric metadimers can be arranged in a one-dimensional lattice in two distinct configurations; one with all the SRRs equidistant and the other with the SRRs forming a PT dimer chain.
Here, we study the case where the electric interaction, that depends on the distance of the gaps, is not neglected, see Fig.\ref{lazaridisdimernew}. In our case we take the gaps looking at the same direction. %Furthermore we investigate two specific cases for different kinds of nonlinearities. % coefficients $a,\beta$.

\begin{figure}[tbp]
\begin{center}
{\includegraphics[width=0.45\textwidth]{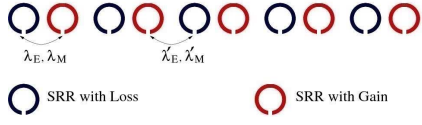}}
\end{center}
\caption{Schematic of a PT-symmetric metamaterial. The separation between SRRs can be modulated according to a binary pattern (i.e., PT dimer chain). The parameters $\lambda _{M}, \lambda' _{M}$ and $\lambda_{E}, \lambda'_{E}$ denote the magnetic (subscript M) and electric (subscript E) interactions between the rings. The sketch is adapted from \cite{laza13}. }\label{lazaridisdimernew}
\end{figure}

In the equivalent circuit model picture (see, e.g., \cite{Rosanov,Lazarides2,Lazarides3,Molina}), extended for the PT dimer chain, the dynamics of the charge $q_{n}$ in the capacitor of the $n$-th SRR is governed by
\begin{equation}\label{eq1}
\begin{gathered}
\ddot{q}_{2n+1}+q_{2n+1}+\lambda_{M} ' \ddot{q}_{2n}+\lambda_{M} \ddot{q}_{2n+2}+
\lambda_{E}'q_{2n}+\lambda_{E} q_{2n+2}\\
=\varepsilon \sin (pn+\Omega t)-a q_{2n+1}
^{2} -\beta q_{2n+1} ^{3}-\gamma\dot{q}_{2n+1},\\
\ddot{q}_{2n+2}+q_{2n+2}+\lambda_{M}  \ddot{q}_{2n+1}+\lambda_{M} ' \ddot{q}_{2n+3}+
\lambda_{E}q_{2n+1}+\lambda_{E}' q_{2n+3}\\
=\varepsilon \sin (pn+\Omega t)-a q_{2n+2}
^{2} -\beta q_{2n+2} ^{3} +\gamma\dot{q}_{2n+2},
\end{gathered}
\end{equation}
where $\lambda _{M}, \lambda' _{M}$ and $\lambda_{E}, \lambda'_{E}$ are the magnetic and electric interaction coefficients, respectively, between nearest
neighbours, $a$ and $\beta$ are nonlinear coefficients, $\gamma$ is the gain or loss coefficient ($\gamma>0$), $\varepsilon$ is the amplitude of the external driving voltage, while $\Omega$ is the driving frequency normalized to $\omega_{0}=1/ \sqrt{LC_{0}}$ and $t$ temporal variable normalized to $\omega^{-1}_{0}$, with $C_{0}$ being the linear capacitance.

Replacing $q_{2n+1}=U_{n}$ and $q_{2n+2}=V_{n}$ into \eqref{eq1}, we obtain
\begin{equation}\label{eq5}
\begin{gathered}
\ddot{U}_{n}+U_{n}+a U_{n} ^{2} +\beta U_{n} ^{3}\\
=-\lambda _{M}' \ddot{V}_{n-1}-\lambda_{M} \ddot{V}_{n}-\lambda_{E} ' 
V_{n-1} -\lambda_{E}V_{n}-\gamma \dot{U}_{n} +\varepsilon \sin (pn+
\Omega t),\\
\ddot{V}_{n}+V_{n}+a V_{n} ^{2} +\beta V_{n} ^{3}\\
=-\lambda _{M} \ddot{U}_{n}-\lambda_{M}' \ddot{U}_{n+1}-\lambda_{E}  
U_{n} -\lambda_{E} ' U_{n+1}+\gamma \dot{V}_{n} +\varepsilon \sin (pn+
\Omega t).
\end{gathered}
\end{equation}
Looking for travelling waves and hence setting $z=pn+\Omega t$, $U(z)=U_n(t)$, $V(z)=V_n(t)$, \eqref{eq5} becomes
\begin{equation}\label{eq7}
\begin{gathered}
\Omega^{2} U_{zz} (z) +U(z)+a U^{2} (z) +\beta U^{3} (z)\\
=-\lambda _{M} ' \Omega^{2} V_{zz} (z-p)-\lambda_{M} \Omega ^{2} V_{zz} (z) -\lambda_{E}' V(z-p)-\lambda_{E} V(z) -\gamma\Omega U_{z} (z) +\varepsilon \sin z,\\
\Omega^{2} V_{zz} (z) +V(z)+a V^{2} (z) +\beta V^{3} (z)\\
=-\lambda _{M} \Omega^{2} U_{zz} (z)-\lambda_{M} '  \Omega ^{2} U_{zz} (z+p) -\lambda_{E} U(z)-\lambda_{E} ' U(z+p) +\gamma\Omega V_{z} (z) +\varepsilon \sin z.
\end{gathered}
\end{equation}
Our work is concerned with the existence and persistence of localized travelling waves of \eqref{eq7}. This paper extends our previous studies \cite{dibl14,feck17}, where the persistence of localized travelling waves in a single equation was considered. In the context of PT-symmetric metamaterials, our work extends that of \cite{laza13} that demonstrated numerically the existence of standing, highly-localized solutions.

This paper is organized as follows. In Section \ref{sec2} we introduce a Melnikov analysis related to homoclinic motion in the system and study two different cases for different values of our nonlinear coefficients $a$ and $\beta$. In Section \ref{sec3}, we present numerical computations comparing the analytical results in the preceding section. Finally, Section \ref{sec4} summarizes our findings.

\section{Melnikov Analysis}
\label{sec2}

In this section we perform a Melnikov analysis for homoclinic type motions of \eqref{eq7} for weak coupling, forcing and damping, which is expressed after scaling $\lambda _{M} \to\varepsilon\lambda _{M} $, $\lambda _{M} '\to\varepsilon\lambda _{M} '$, $\lambda _{E} \to\varepsilon\lambda _{E} $, $\lambda _{E} '\to\varepsilon\lambda _{E} '$, and $z\to z/\Omega$ as
\begin{equation}\label{eq10a}
\begin{gathered}
U_{zz} (z) +U(z)+a U^{2} (z) +\beta U^{3} (z)\\
=\varepsilon(-\lambda _{M} 'V_{zz} (z-p)-\lambda_{M}V_{zz} (z) -\lambda_{E}' V(z-p)-\lambda_{E} V(z) -\gamma U_{z} (z) +\sin \Omega z),\\
V_{zz} (z) +V(z)+a V^{2} (z) +\beta V^{3} (z)\\
=\varepsilon(-\lambda _{M}U_{zz} (z)-\lambda_{M} 'U_{zz} (z+p) -\lambda_{E} U(z)-\lambda_{E} ' U(z+p) +\gamma V_{z} (z) +\sin \Omega z),
\end{gathered}
\end{equation}
where the parameter $|\varepsilon| \ll 1$ indicates the pertubative character in the above equations. For $\varepsilon=0$ the unperturbed system is
\begin{equation}\label{eq9}
\begin{gathered}
U_{zz} (z) +U(z)+a U^{2} (z) +\beta U^{3} (z) =0,\\
V_{zz} (z) +V(z)+a V^{2} (z) +\beta V^{3} (z) =0.
\end{gathered}
\end{equation}
We consider the following cases: a) $a<0$ and $\beta=0$, and b) $a=0$ and $\beta=-1$, which represent systems with non-topological and topological localized waves.

\subsection{Case (a)}

Both equations of \eqref{eq9} have a hyperbolic equilibrium $(\overline{p}_{i},\overline{q}_{i})=\left(0,-\frac{1}{a}\right)$, $i=1,2$ where $i=1$ and $i=2$ refer to the first and second equation of \eqref{eq9}, $q_1=U$, $q_2=V$, connected by a homoclinic solution
$$p_i=\dot{q_{i}},\quad q_{i}(z)=-\frac{1}{a}+\frac{3\sech ^{2}\frac{z}{2}}{2a}.$$
So in the full space $\R^4$, the system \eqref{eq9} has a hyperbolic equilibrium
$$(\overline{p}_{1},\overline{q}_{1},\overline{p}_{2},\overline{q}_{2})$$
connected by the homoclinic trajectory
\begin{equation}\label{homoclinic trajectories}
\begin{gathered}
(p_{1h} (z),q_{1h} (z),p_{2h} (z),q_{2h} (z))\\
=\left(-\frac{3\sech^{2}\frac{z}{2}\tanh\frac{z}{2}}{2a},-\frac{1}{a}+\frac{3\sech^{2} \frac{z}{2}}{2a},-\frac{3\sech^{2}\frac{z}{2}\tanh\frac{z}{2}}{2a},-\frac{1}{a}+\frac{3\sech^{2}\frac{z}{2}}{2a}\right).
\end{gathered}
\end{equation}
To study the persistence of homoclinic type solutions for \eqref{eq10a}, we compute the Melnikov integrals (for higher dimensions Melnikov analysis see \cite{Wiggins1}, \cite{Wiggins2}, \cite{Yagasaki1} and \cite{Yagasaki2}). After
introducing real parameters $\tau_{1}, \tau_{2}$ determining the position on the homoclinic type orbits, we derive
$$
\begin{gathered}
M_{1} (\tau _{1},\tau_{2}):=\int_{-\infty}^{\infty} p_{1h} (z-\tau _{1})[-\lambda_{M}' \dot{p}_{2h} (z-\tau_{2} -p)-\lambda_{M} \dot{p}_{2h} (z-\tau_{2})\\
-\lambda_{E}' q_{2h} (z-\tau_{2}-p)-\lambda_{E} q_{2h} (z-\tau_{2})-\gamma p_{1h} (z-\tau_{1})+\sin \Omega z]dz,\\
M_{2} (\tau _{1},\tau_{2}):=\int_{-\infty}^{\infty} p_{2h} (z-\tau_{2}) [-\lambda_{M} \dot{p}_{1h} (z-\tau_{1})-\lambda_{M}'\dot{p}_{1h} (z-\tau_{1}+p)\\
-\lambda_{E}q_{1h} (z-\tau_{1})-\lambda_{E}' q_{1h} (z-\tau_{1}+p)+\gamma p_{2h} (z-\tau_{2})+\sin \Omega z]dz.
\end{gathered}$$
We see that the above Melnikov functions have the forms
\begin{equation}
\begin{gathered}
M_{1}(\tau _{1},\tau_{2})=M_{11}(\tau _{1}-\tau_{2})+M_{12}\cos\Omega \tau_1,\\
M_{2}(\tau _{1},\tau_{2})=M_{21}(\tau _{1}-\tau_{2})+M_{22}\cos\Omega \tau_2,
\end{gathered}
\label{ins1}
\end{equation}
where $M_{i1}$ are analytic functions with $\lim_{\eta\to\pm\infty}M_{i1}(\eta)=M_{i1\pm}$ exist and $M_{i2}$ are nonzero constants (see Appendices A and B). Now we pass to $\eta=\tau_1-\tau_2$ and $\tau_2$, so we are looking for a simple zero of
\begin{equation}\label{e1}
\begin{gathered}
0=M_{11}(\eta)+M_{12}\cos\Omega(\eta+\tau_2)=M_{11}(\eta)+M_{12}(\cos\Omega\eta\cos\Omega\tau_2-\sin\Omega\eta\sin\Omega\tau_2),\\
M_{21}(\eta)+M_{22}\cos\Omega \tau_2=0.
\end{gathered}
\end{equation}
For $\sin\Omega\eta\ne0$, solving \eqref{e1}, we obtain
\begin{equation}\label{e2}
\begin{gathered}
\cos\Omega \tau_2=-\frac{M_{21}(\eta)}{M_{22}},\\
\sin\Omega\tau_2=\frac{M_{11}(\eta)M_{22}-M_{12}M_{21}(\eta)\cos\Omega\eta}{M_{12}M_{22}\sin\Omega\eta}.\end{gathered}
\end{equation}
By $\cos^2\Omega \tau_2+\sin^2\Omega \tau_2=1$, this yields
\begin{equation}\label{e3}
\begin{gathered}
M(\eta):=M_{21}^2(\eta)M_{12}^2\sin^2\Omega\eta+\left(M_{11}(\eta)M_{22}-M_{12}M_{21}(\eta)\cos\Omega\eta\right)^2\\-M_{12}^2M_{22}^2\sin^2\Omega\eta=0.
\end{gathered}
\end{equation}
For $\sin\Omega\eta=0$, i.e., $\eta_k=\frac{\pi k}{\Omega}$ for some $k\in \mathbb{Z}$, solving \eqref{e1}, we obtain
\begin{equation}\label{e2b}
\cos\Omega \tau_2=-\frac{M_{21}\left(\frac{\pi k}{\Omega}\right)}{M_{22}}
\end{equation}
and
\begin{equation}\label{e2c}
M_{11}\left(\frac{\pi k}{\Omega}\right)M_{22}=(-1)^kM_{12}M_{21}\left(\frac{\pi k}{\Omega}\right).
\end{equation}
\begin{lemma}
A simple root $\eta$ of \eqref{e3} with $\eta\ne\frac{\pi k}{\Omega}$ for any $k\in \mathbb{Z}$ gives via \eqref{e2} a simple zero of \eqref{e1}.
\end{lemma}
\begin{proof}
Clearly simple zeroes of \eqref{e2} are equivalent to simple zeroes of \eqref{e1} when $\sin\Omega\eta\ne0$. Next, we can write \eqref{e2} as
\begin{equation}\label{e2ab}
\begin{gathered}
\cos\Omega \tau_2-A(\eta)=0,\\
\sin\Omega\tau_2-B(\eta)=0\end{gathered}
\end{equation}
for $A(\eta)=\frac{\widetilde{M}_{1}(\eta)}{M_{3}(\eta)}$, $B(\eta)=\frac{\widetilde{M}_{2}(\eta)}{M_{3}(\eta)}$, $M_3(\eta)=M_{12}M_{22}\sin\Omega\eta$ and others defined by definition. Then, by \eqref{e3}, $M(\eta)=\left(A^2(\eta)+B^2(\eta)-1\right)M_3^2(\eta)$. The Jacobian of \eqref{e2ab} at its zero $\eta_0$ is
$$
\begin{gathered}
\Omega\left(A(\eta_0)'\cos\Omega\tau_2+B(\eta_0)'\sin\Omega\tau_2\right)=\Omega\left(A(\eta_0)'A(\eta_0)+B(\eta_0)'B(\eta_0)\right)\\
=\frac{\Omega}{2}\left(A^2(\eta_0)+B^2(\eta_0)\right)'=\frac{\Omega}{2}\left(\frac{M(\eta_0)}{M_3^2(\eta_0)}\right)'=\frac{\Omega M(\eta_0)'}{2M_3^2(\eta_0)},
\end{gathered}
$$
since $M(\eta_0)=0$. Hence any zero of \eqref{e2} is simple if and only if it is generated by a simple root of \eqref{e3}. The proof is finished.
\end{proof}

Next, \eqref{e3} is asymptotically near at $\pm\infty$ to
\begin{equation}\label{e4}
\begin{gathered}
M_\pm(\eta):=M_{21\pm}^2M_{12}^2\sin^2\Omega\eta+\left(M_{11\pm}M_{22}-M_{12}M_{21\pm}\cos\Omega\eta\right)^2-M_{12}^2M_{22}^2\sin^2\Omega\eta\\
=M_{21\pm}^2M_{12}^2+M_{11\pm}^2M_{22}^2-M_{12}^2M_{22}^2-2M_{11\pm}M_{22}M_{12}M_{21\pm}\cos\Omega\eta\\
+M_{12}^2M_{22}^2\cos^2\Omega\eta=0,
\end{gathered}
\end{equation}
where 
$$
\begin{gathered}
M_{11\pm}=-\int_{-\infty}^{\infty} \frac{9\gamma}{4a^{2}} \sech^{4}\frac{z}{2}\tanh^{2}\frac{z}{2}dz=-\frac{6\gamma}{5a^2},\\
M_{21\pm}=\int_{-\infty}^{\infty} \frac{9\gamma}{4a^{2}} \sech^{4}\frac{z}{2}\tanh^{2}\frac{z}{2}dz=\frac{6\gamma}{5a^2}.
\end{gathered}
$$

Equation \eqref{e4} has roots
\begin{equation}\label{e5}
\cos\Omega\eta=\frac{-\widetilde{b}\pm\sqrt{\widetilde{b}^2-4\widetilde{a}\widetilde{c}}}{2\widetilde{a}}
\end{equation}
for
$$
\widetilde{a}=M_{12}^2M_{22}^2,\quad \widetilde{b}=-2M_{11\pm}M_{22}M_{12}M_{21\pm},\quad \widetilde{c}=M_{21\pm}^2M_{12}^2+M_{11\pm}^2M_{22}^2-M_{12}^2M_{22}^2.
$$
When
\begin{equation*}
\left|\frac{-\widetilde{b}\pm\sqrt{\widetilde{b}^2-4\widetilde{a}\widetilde{c}}}{2\widetilde{a}}\right|<1,
\end{equation*}
i.e.,
\begin{equation}\label{e6.0}
\Big|\gamma^2(\cosh^2\Omega\pi-1)\pm\left(\gamma^2(\cosh^2\Omega\pi-1)-25a^2\Omega^4\pi^2\right)\Big|<25a^2\Omega^4\pi^2
\end{equation}
for either sign, then \eqref{e5} and thus also \eqref{e4} has infinitely many simple roots which gives large (in absolute value) simple roots of \eqref{e3} with $|\cos\Omega\eta|<1$, i.e. $\sin\Omega\eta\ne0$. For such large $\eta$ via \eqref{e2}, we obtain a simple zero of \eqref{e1}.
Looking at \eqref{e6.0}, one can see that the inequality is never satisfied for the minus sign. Thus it is equivalent to
\begin{equation}\label{e6}
	\gamma\,|\sinh\Omega\pi|<-5a\Omega^2\pi.
\end{equation}
Summarizing we obtain the following result.

\begin{theorem} Let $a<0$ and $\beta=0$. Condition \eqref{e6} is sufficient for the persistence of a homoclinic type solution in \eqref{eq10a} for $\varepsilon\ne0$ small.
\end{theorem}

Finally we do not study the case $\sin\Omega\eta=0$ in more details, since then \eqref{e2c} must hold which is rather strong restriction on parameters occurring in \eqref{eq10a}.

\begin{figure}[tbp]
\begin{center}
\subfigure[]{\includegraphics[width=0.45\textwidth]{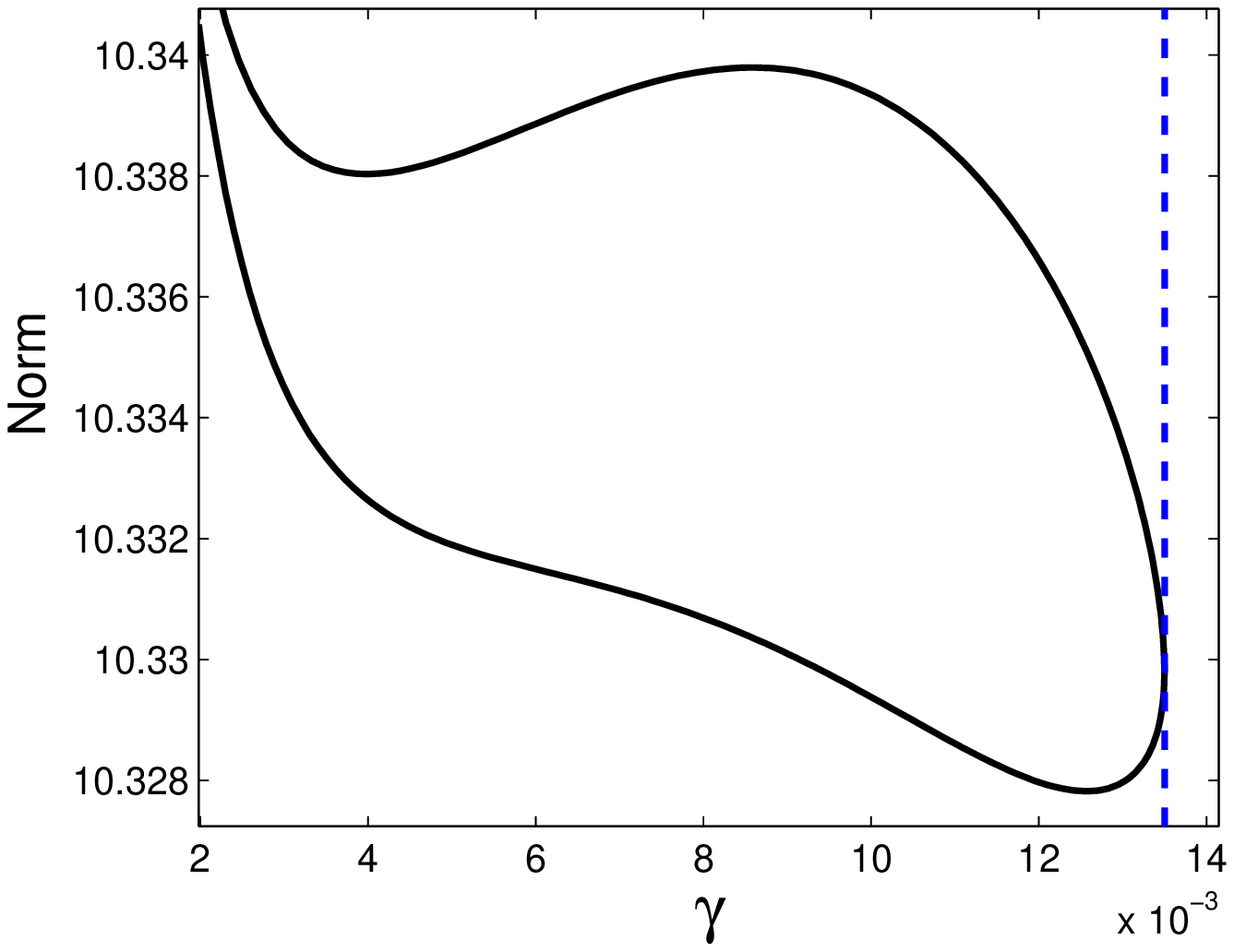}}
\subfigure[]{\includegraphics[width=0.45\textwidth]{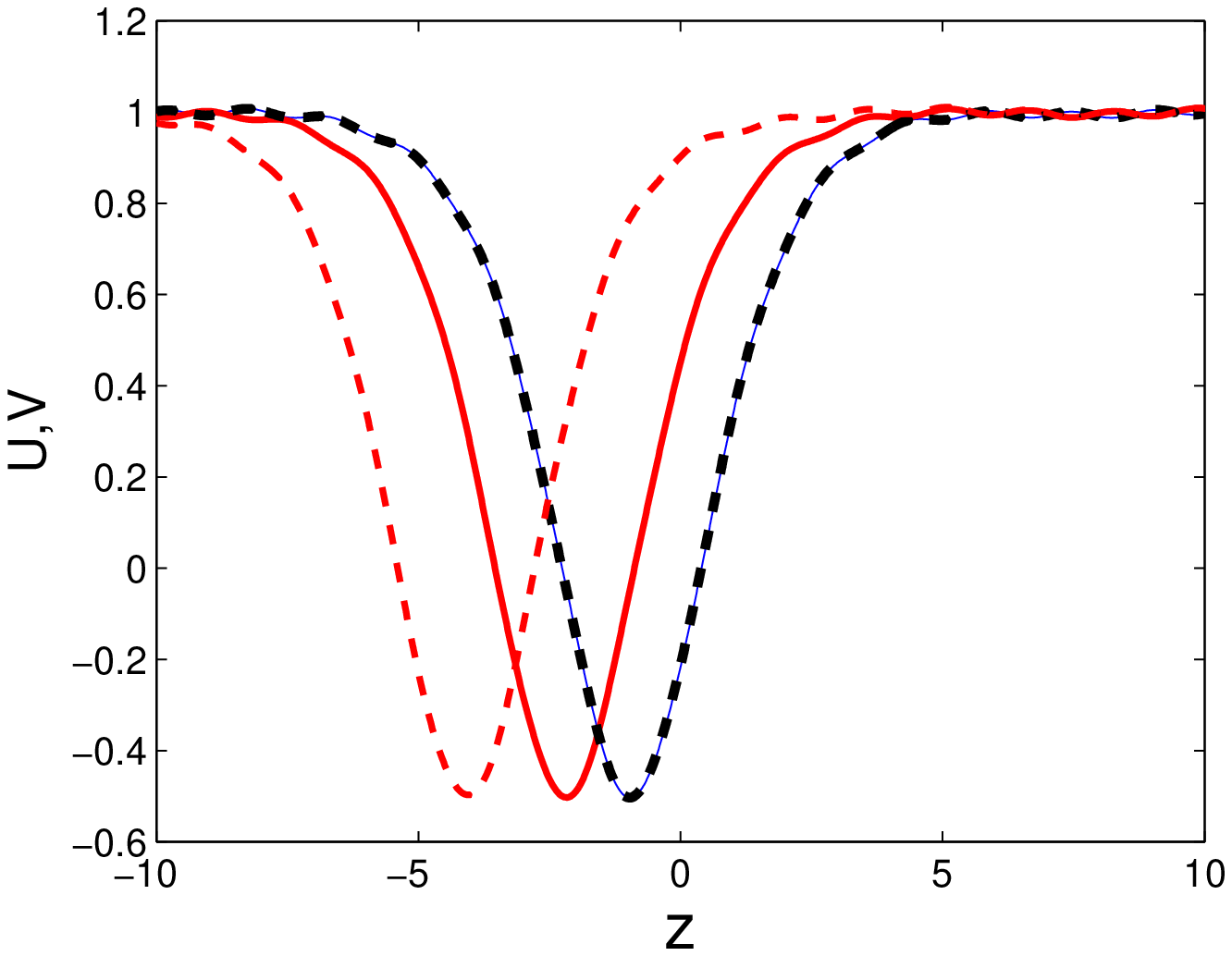}}\\
\subfigure[]{\includegraphics[width=0.45\textwidth]{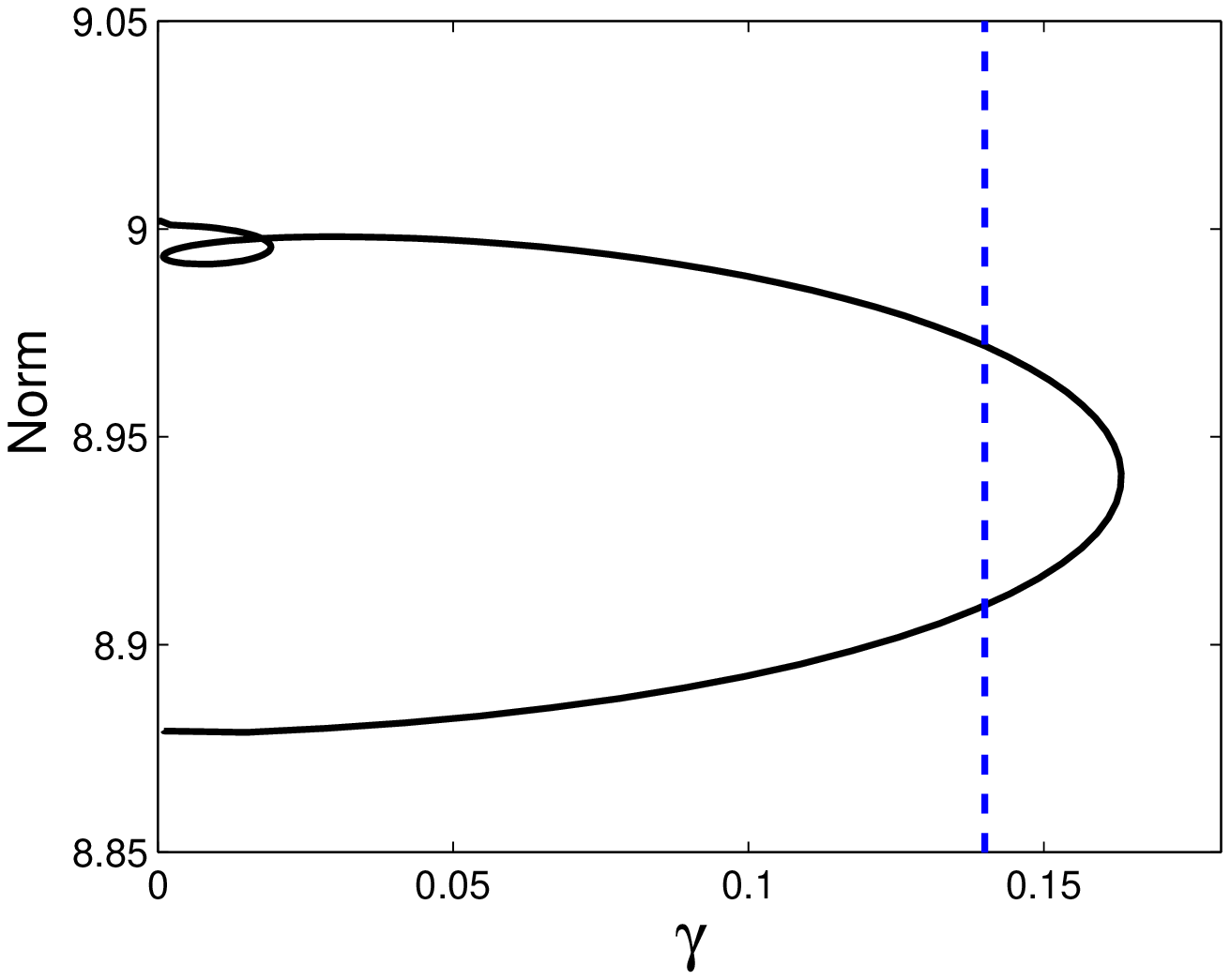}}
\subfigure[]{\includegraphics[width=0.45\textwidth]{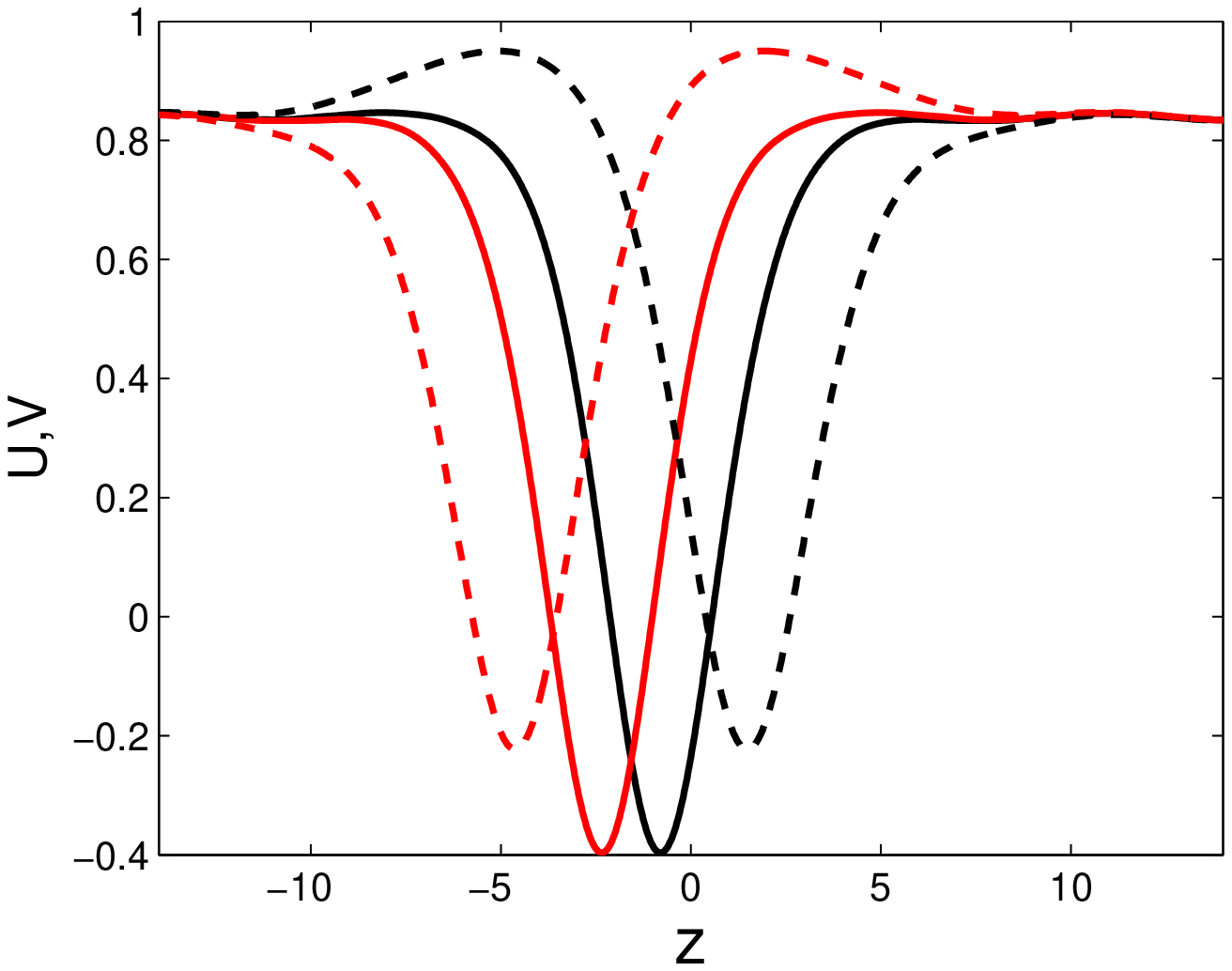}}
\end{center}
\caption{(a) Bifurcation diagram of the localized solution for $a=-1$, $\beta=0$, $\lambda_M'=\lambda_M=\lambda_E'=\lambda_E=0$ and $\varepsilon=0.01.$ The vertical axis is the solution norm and the horizontal axis is the gain-loss coefficient $\gamma$. Panel (b) shows the solutions on the lower (solid line) and upper (dashed line) branches at $\gamma=0.008$. The right and left dips correspond to $U$ and $V$, respectively. Panels (c,d) are the same as (a,b), but for $\lambda'_{M}=-0.020$, $\lambda'_{E}=-0.040$, $\lambda_{M}=-0.040$, $\lambda_{E}=-0.120$. Panel (d) shows the solutions at $\gamma=0.1.$ The vertical dashed lines are the existence boundaries given by the Melnikov function.}\label{pic1}
\end{figure}

\subsection{Case (b)}
The equations of \eqref{eq9} possess heteroclinic solutions
$$p_i=\dot q_i,\quad q_{i}(z)=\tanh\frac{z}{\sqrt{2}}.$$
So in the full space $\R^4$, the system \eqref{eq9} has hyperbolic equilibria
$$(\overline{p}_{1}^\pm,\overline{q}_{1}^\pm,\overline{p}_{2}^\pm,\overline{q}_{2}^\pm)=(0,\pm1,0,\pm1)$$
connected by the heteroclinic trajectory
\begin{equation}\label{heteroclinc trajectories}
(p_{1h} (z),q_{1h} (z),p_{2h} (z),q_{2h} (z))
=\left(\frac{\sech^{2} \frac{z}{\sqrt{2}}}{\sqrt{2}},\tanh\frac{z}{\sqrt{2}},\frac{\sech^{2} \frac{z}{\sqrt{2}}}{\sqrt{2}},\tanh\frac{z}{\sqrt{2}}\right).
\end{equation}
The Melnikov integrals are
\begin{equation}
\begin{gathered}
M_{1}(\tau _{1},\tau_{2})=N_{11}(\tau _{2}-\tau_{1})+N_{12}\sin\Omega \tau_1,\\
M_{2}(\tau _{1},\tau_{2})=N_{21}(\tau _{2}-\tau_{1})+N_{22}\sin\Omega \tau_2,
\end{gathered}
\label{ins2}
\end{equation}
where $N_{i1}$ are analytic functions with $\lim_{\eta\to\pm\infty}N_{i1}(\eta)=N_{i1\pm}$ existing and $N_{i2}$ are nonzero constants (see Appendix B). Now taking $\tau_i=\frac{\pi}{2\Omega}-\xi_i$, $i=1,2$, we derive
$$
\begin{gathered}
N_{1}(\xi_{1},\xi_{2})=M_{1}\left(\frac{\pi}{2\Omega}-\xi_{1},\frac{\pi}{2\Omega}-\xi_{2}\right)=N_{11}(\xi_{1}-\xi_{2})+N_{12}\cos\Omega \xi_1,\\
N_{2}(\xi_{1},\xi_{2})=M_{2}\left(\frac{\pi}{2\Omega}-\xi_{1},\frac{\pi}{2\Omega}-\xi_{2}\right)=N_{21}(\xi_{1}-\xi_{2})+N_{22}\cos\Omega \xi_2,
\end{gathered}
$$
so we can directly apply arguments of case (a) to derive the following result. Of course, instead of \eqref{e3}, now we have the function
\begin{equation}\label{e3b}
\begin{gathered}
N(\eta)=N_{21}^2(\eta)N_{12}^2\sin^2\Omega\eta+\left(N_{11}(\eta)N_{22}-N_{12}N_{21}(\eta)\cos\Omega\eta\right)^2\\
-N_{12}^2N_{22}^2\sin^2\Omega\eta=0.
\end{gathered}
\end{equation}
Analogously, we derive the condition
\begin{equation}\label{e6b.0}
\left|\frac{-\bar{b}\pm\sqrt{\bar{b}^2-4\bar{a}\bar{c}}}{2\bar{a}}\right|<1
\end{equation}
with
\begin{equation*}
\bar{a}=N_{12}^2N_{22}^2,\quad \bar{b}=-2N_{11\pm}N_{22}N_{12}N_{21\pm},\quad \bar{c}=N_{21\pm}^2N_{12}^2+N_{11\pm}^2N_{22}^2-N_{12}^2N_{22}^2.
\end{equation*}
Putting $\bar{a}$, $\bar{b}$, $\bar{c}$ into \eqref{e6b.0}, one can see that it is equivalent to
\begin{equation*}
	\left|\pi^2\Omega^2\csch^2\frac{\Omega\pi}{\sqrt{2}}-4\left(\pm(\lambda_E+\lambda_E')-\frac{\sqrt{2}\gamma}{3}\right)^2\right|<\pi^2\Omega^2\csch^2\frac{\Omega\pi}{\sqrt{2}},
\end{equation*}
i.e.,
\begin{equation}\label{e6b}
	0<\left|\pm(\lambda_E+\lambda_E')-\frac{\sqrt{2}\gamma}{3}\right|
	<\frac{\Omega\pi}{\sqrt{2}}\csch\frac{\Omega\pi}{\sqrt{2}}.
\end{equation}
\begin{theorem} Let $a=0$ and $\beta=-1$. Condition \eqref{e6b} for either sign is sufficient for the persistence of a heteroclinic type solution in \eqref{eq10a} for $\varepsilon\ne0$ small.
\end{theorem}

\section{Numerical results}
\label{sec3}

\begin{figure}[tbhp]
\begin{center}
\subfigure[]{\includegraphics[width=0.45\textwidth]{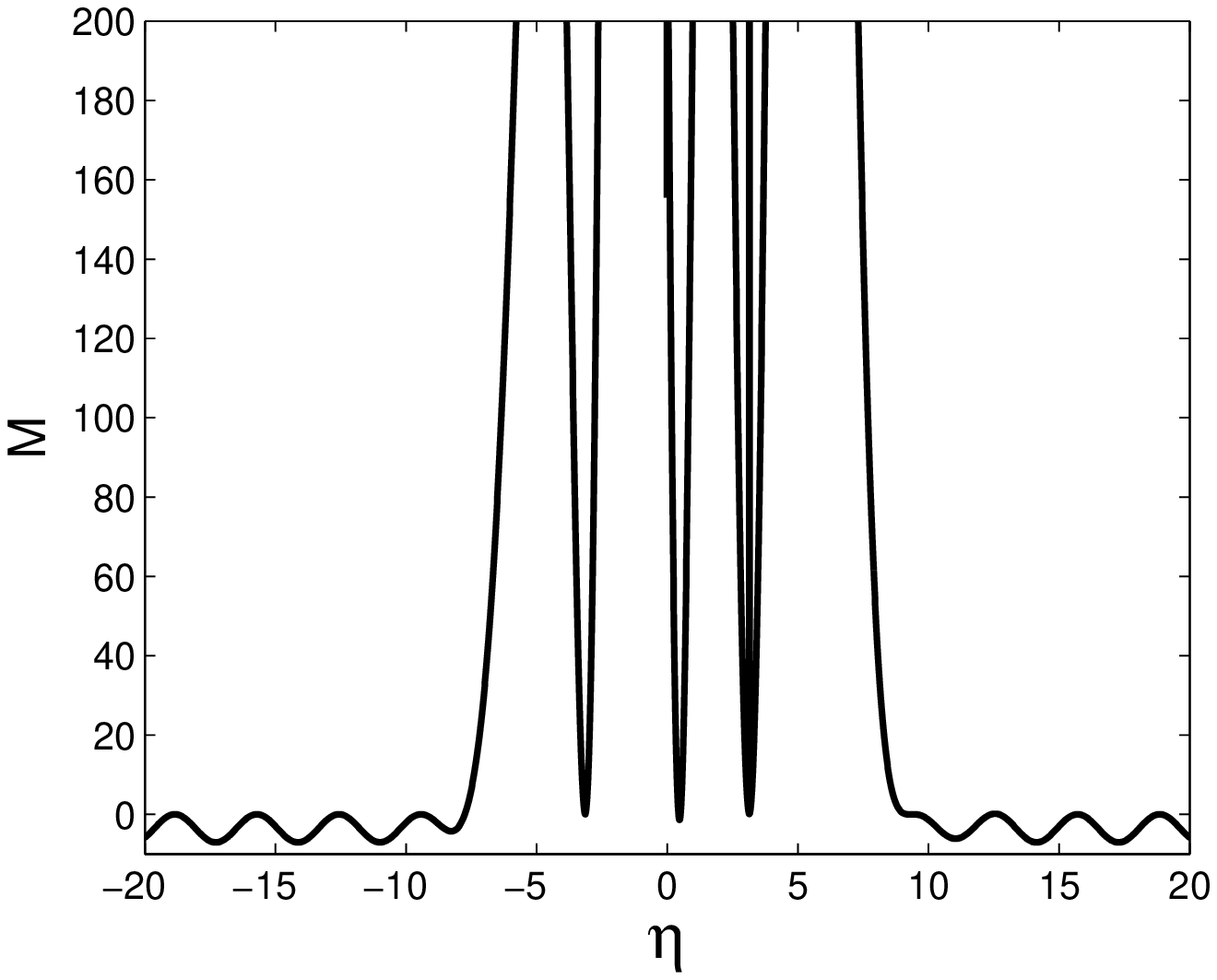}}
\subfigure[]{\includegraphics[width=0.45\textwidth]{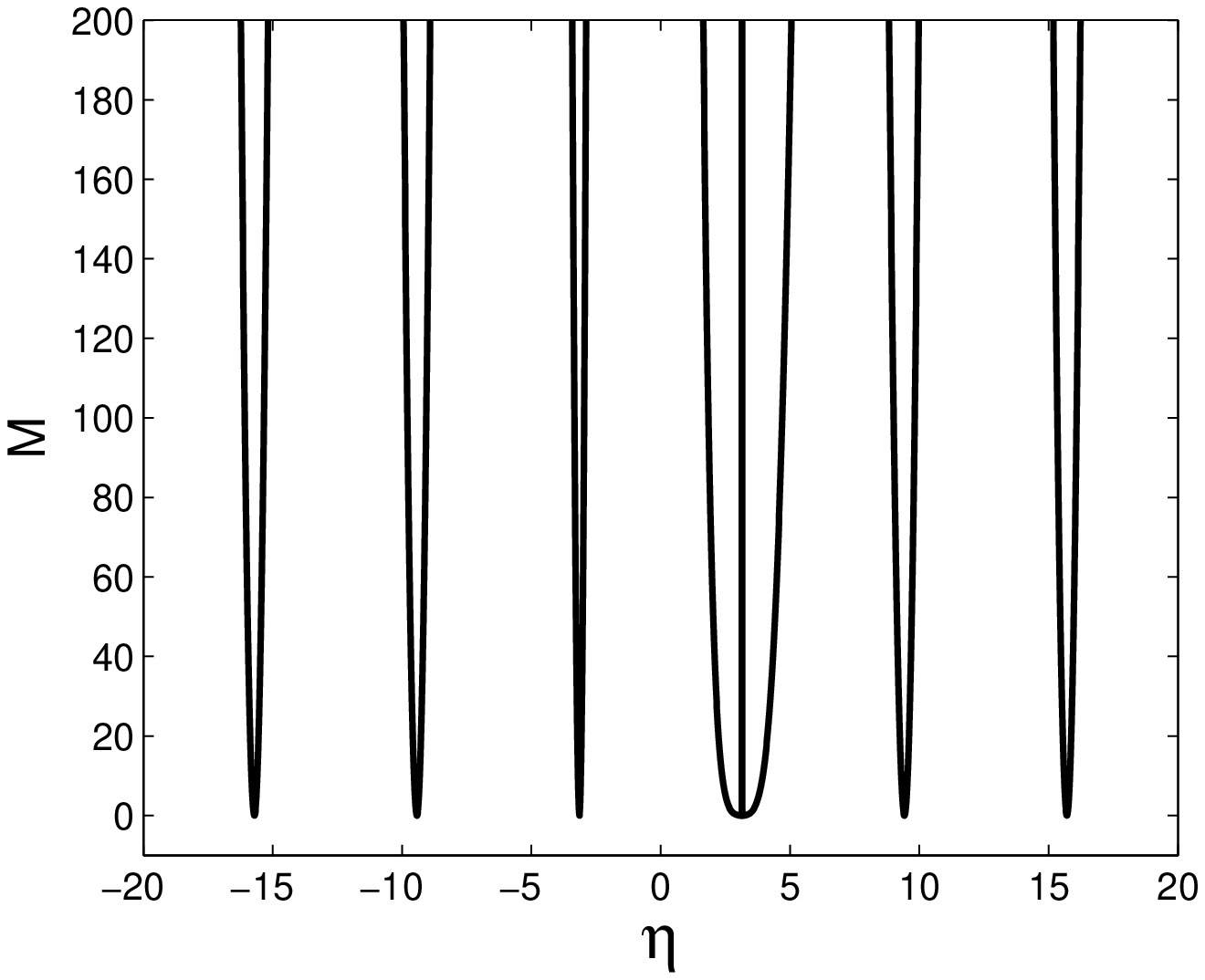}}
\end{center}
\caption{The Melnikov function $M$ of \eqref{e3} with $\lambda'_{M}=-0.020/\varepsilon$, $\lambda'_{E}=-0.040/\varepsilon$, $\lambda_{M}=-0.040/\varepsilon$, $\lambda_{E}=-0.120/\varepsilon$ and $\gamma=0$ (a) or $\gamma=0.14/\varepsilon$ (b). Here, $\varepsilon=0.01.$}\label{pic2}
\end{figure}

To illustrate the theoretical results obtained in the previous sections, we have solved the governing equations \eqref{eq7} numerically. The advance-delay equation \eqref{eq7} is solved using a pseudo-spectral method, i.e., we express the solutions $U$ and $V$ in the trigonometric polynomials
\begin{equation}
\begin{gathered}
U(z)=\sum_{j=1}^{J}\left[A_j\cos\left((j-1)\tilde{k}z\right)+B_j\sin\left(j\tilde{k}z\right)\right],\\
V(z)=\sum_{j=1}^{J}\left[C_j\cos\left((j-1)\tilde{k}z\right)+D_j\sin\left(j\tilde{k}z\right)\right],
\label{ser}
\end{gathered}
\end{equation}
where $\tilde{k}=2\pi/L$ and $-L/2<z<L/2$. Substituting the series into 
\eqref{eq7} and requiring the representations to satisfy the equations 
at $2J$ collocation points, we obtain a system of algebraic equations 
for the Fourier coefficients $A_j$, $B_j$, $C_j$ and $D_j$, $j=1,2,
\dots,J,$ which are then solved using Newton's method. Typically, we 
use $L=30\pi$ and $J=45$. In the presence of coupling, using the 
experimentally relevant parameter values provided in \cite{Rosanov}, in 
the following we take $\lambda'_{M}=-0.020$, $\lambda'_{E}=-0.040$, $
\lambda_{M}=-0.040$ and $\lambda_{E}=-0.120$.

First, we consider case (a) by taking $a=-1$ and $\beta=0$. Shown in Fig.~\ref{pic1} is the bifurcation diagram of the localized solution \eqref{homoclinic trajectories} for two sets of coupling constants. The vertical axis is the norm defined as
\[
Norm = \sqrt{\int_{-L/2}^{L/2}|U(z)|^2+|V(z)|^2\,dz}.
\]
 By fixing the driving amplitude $\varepsilon$ and varying the gain-loss coefficient $\gamma$, we obtain a fold bifurcation. In the figure, we also present the existence boundary, i.e., the fold bifurcation, predicted by our Melnikov function analysis, which is obtained as the following.

In Fig.~\ref{pic2}, we plot the function $M$ of \eqref{e3} for two different values of gain-loss parameter, i.e.\ $\gamma=0$ and $\gamma=0.14/\varepsilon$, with $\lambda'_{M}=-0.020/\varepsilon$, $\lambda'_{E}=-0.040/\varepsilon$, $\lambda_{M}=-0.040/\varepsilon$, $\lambda_{E}=-0.120/\varepsilon$. The division by $\varepsilon$ is due to the scaling taken in transforming \eqref{eq7} into \eqref{eq10a}. When $\gamma=0$, we see that the function has many roots. As $\gamma$ increases, there is a critical value where all the roots become degenerate as shown in panel (b). In Fig.~\ref{pic1}, the critical value (multiplied by $\varepsilon$ due to the scaling) is depicted as the vertical dashed line, which well predicts the numerical results.

\begin{figure}[tbhp]
\begin{center}
\subfigure[]{\includegraphics[width=0.45\textwidth]{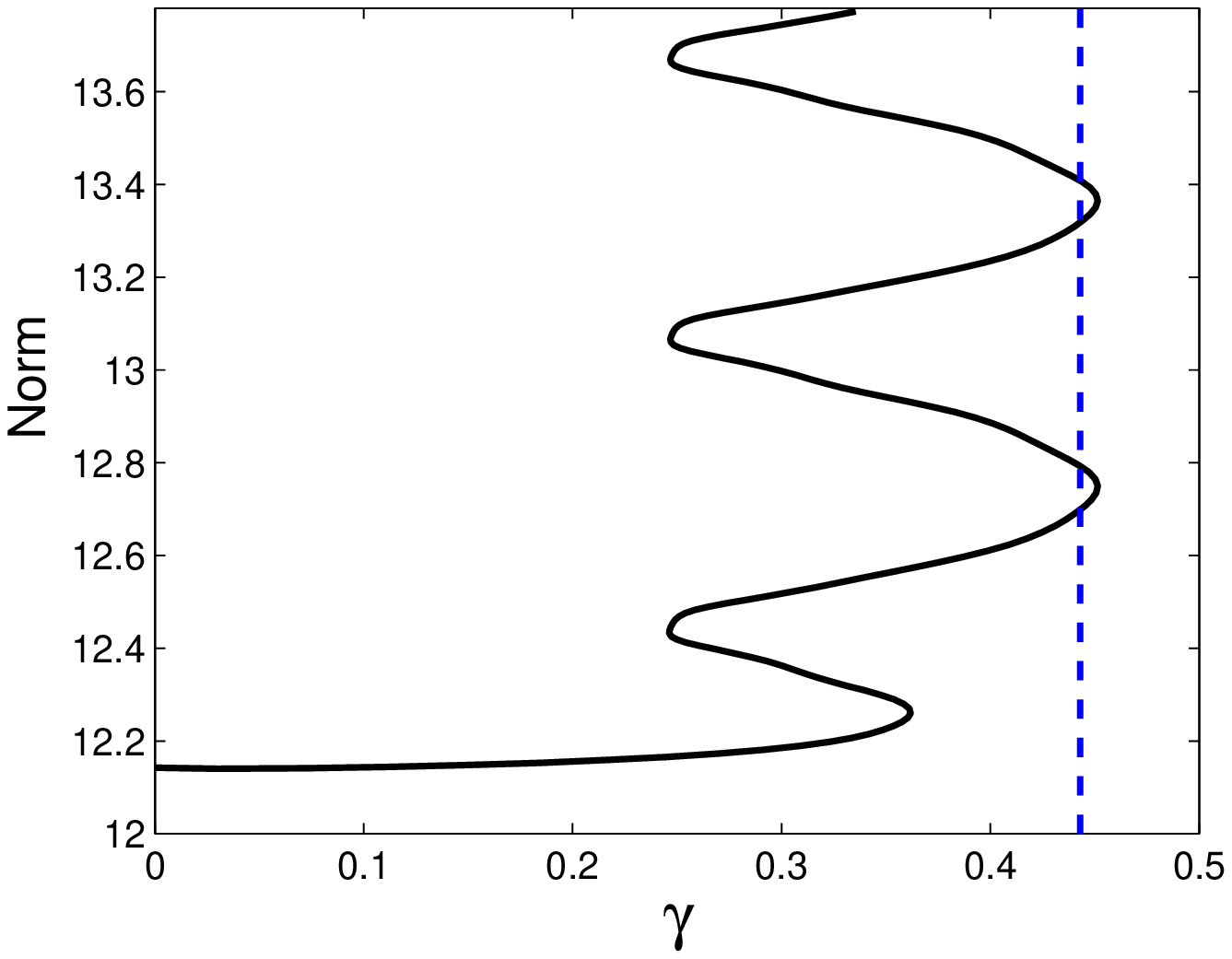}}
\subfigure[]{\includegraphics[width=0.45\textwidth]{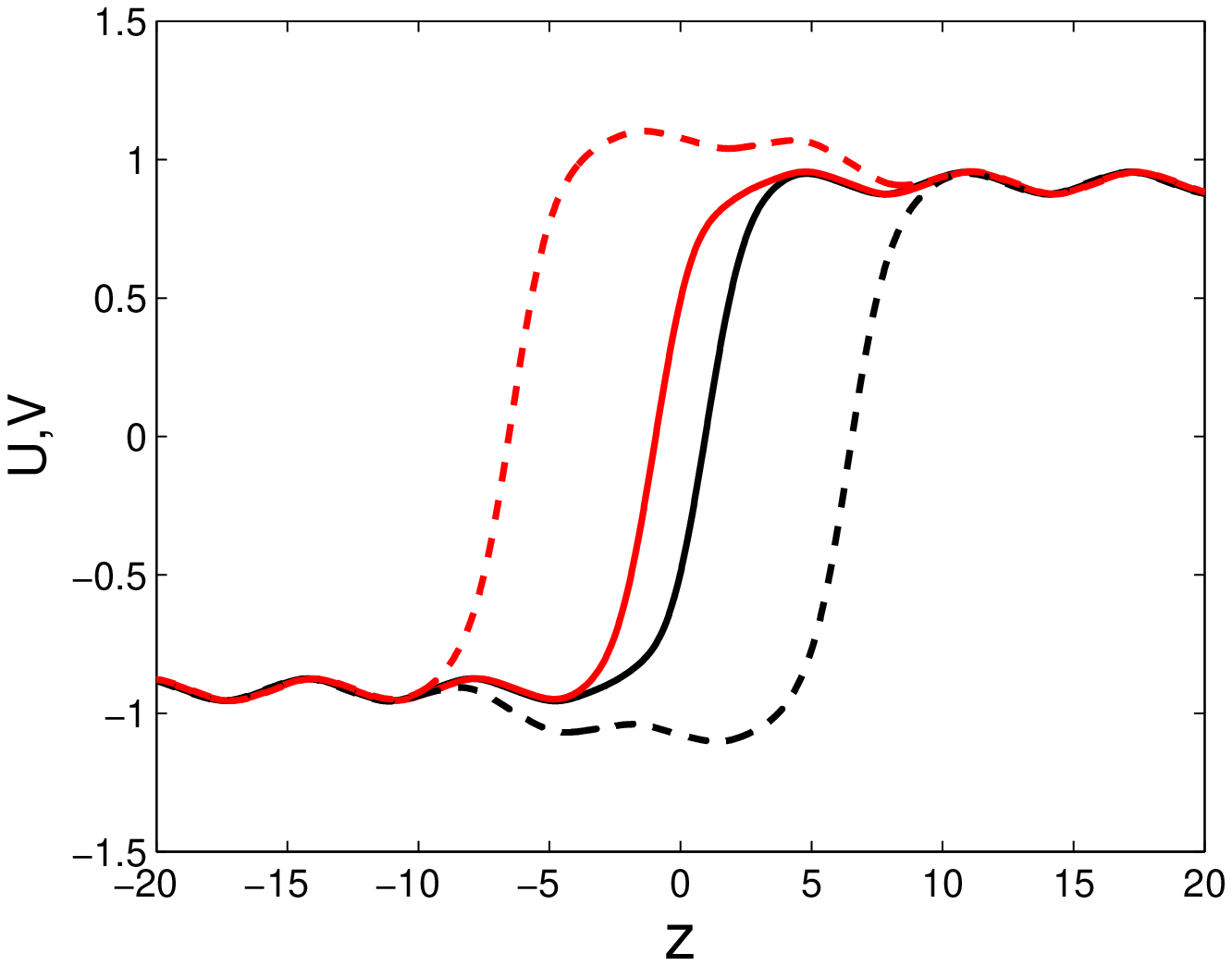}}
\end{center}
\caption{The same as Fig.~\ref{pic1}, but for $a=0$, $\beta=-1$ and $\varepsilon=0.1.$ Panel (b) shows the solutions for $\gamma=0.2$ and $\gamma=0.4$ (at the lowest branch). The right and left kink are $U$ and $V$, respectively. The vertical dashed line in panel (a) is the approximation obtained from the Melnikov function \eqref{e3b}.}\label{pic3}
\end{figure}

We also consider case (b) by taking $a=0$ and $\beta=-1$. Shown in Fig.~\ref{pic3} are the bifurcation diagram and the solution profiles. Different from case (a) above, interestingly as the gain-loss coefficient $\gamma$ varies we obtain a ``snaking'' profile in the bifurcation diagram, where there are many turning points. This phenomenon is now referred to as homoclinic snaking (see, e.g., \cite{matt11} and references therein), as it involves homoclinic or heteroclinic structures. In our current case, along the bifurcation diagram, the distance between the kinks increases. The snaking profile appears because of the locking mechanism between the kinks mediated by the oscillating tail that acts as a periodic (i.e.\ lattice) potential. It is important to note that our Melnikov function \eqref{e3b} could predict well the right boundary of the snaking curve. The persistence of the homoclinic snaking as well as the stability of solutions along the bifurcation diagram are beyond the scope of the present paper and will be reported elsewhere.

\section{Conclusion}
\label{sec4}

We have investigated localized travelling waves in a chain of PT-symmetric nonlinear metamaterials with gain and loss. We have developed a Melnikov function analysis for detecting the persistence of such waves in the system. By considering two cases of nonlinearity, we compared the theoretical analysis and computational results, where good agreement is obtained. One of the considered cases admits localized waves with  a bifurcating diagram exhibiting a ``snaking'' profile.

The existence and stability of periodic travelling waves (following and extending \cite{feck13,agao16} that dealt with only a single advance-delay differential equation), that may be more relevant from applications point of view than travelling localized waves considered herein that most likely are unstable, are addressed for future work.

\section*{Acknowledgments}
%\noindent This part should come before References. Funding information may also be included here.
The work of M.A has been co-financed from resources of the operational program ``Education and Lifelong Learning'' of the European Social Fund and the National Strategic Reference Framework (NSRF) 2007-2013 within the framework of the Action State Scholarships Foundation's (IKY) mobility grants programme for the short term training in recognized scientific/research centres abroad for candidate doctoral or postdoctoral researchers in Greek universities or research centres. M.F. and M.P. were partially supported by the Slovak Grant Agency VEGA No. 2/0153/16 and No. 1/0078/17 and the Slovak Research and Development Agency under the contract  No. APVV-14-0378. VR has been co-financed by the the European Union (European Social Fund ESF) and Greek national funds through the Operational Program Education and Lifelong Learning of the National Strategic Reference Framework (NSRF) Research Funding Program D.534 MIS:379337: THALES. V.R and M.A  acknowledge support from FP7, Marie Curie Actions, People, International Research Staff Exchange Scheme (IRSES-606096). The authors acknowledge the referee for their useful comment and feedback.

\section*{Appendix A}
The terms in the Melnikov function \eqref{ins1} are given by the following integrals
$$
\begin{gathered}
M_{11} (\eta)=-\int_{-\infty}^{\infty} \frac{9\lambda_{M}'}{8a^2} \sech^{2}\frac{z}{2}\tanh\frac{z}{2}\sech^{4} \frac{z+\eta-p}{2}dz\\
+\int_{-\infty}^{\infty} \frac{9\lambda_{M}'}{4a^{2}} \sech^{2}\frac{z}{2}\tanh\frac{z}{2}\sech^{2} \frac{z+\eta-p}{2}\tanh^{2}\frac{z+\eta-p}{2}dz\\
-\int_{-\infty}^{\infty} \frac{9\lambda_{M}}{8a^2} \sech^{2}\frac{z}{2}\tanh\frac{z}{2}\sech^{4}\frac{z+\eta}{2}dz\\
+\int_{-\infty}^{\infty} \frac{9\lambda_{M}}{4a^{2}} \sech^{2}\frac{z}{2}\tanh\frac{z}{2}\sech^{2} \frac{z+\eta}{2}\tanh^{2}\frac{z+\eta}{2}dz\\
-\int_{-\infty}^{\infty} \frac{3\lambda_{E}'}{2a^{2}} \sech^{2}\frac{z}{2}\tanh\frac{z}{2}dz
+\int_{-\infty}^{\infty} \frac{9\lambda_{E}'}{4a^{2}}\sech^{2}\frac{z}{2}\tanh\frac{z}{2}\sech^{2}\frac{z+\eta-p}{2}dz\\
-\int_{-\infty}^{\infty} \frac{3\lambda_{E}}{2a^{2}} \sech^{2}\frac{z}{2}\tanh\frac{z}{2}dz
+\int_{-\infty}^{\infty} \frac{9\lambda_{E}}{4a^{2}} \sech^{2}\frac{z}{2}\tanh\frac{z}{2}\sech^{2}\frac{z+\eta}{2}dz\\
-\int_{-\infty}^{\infty} \frac{9\gamma}{4a^{2}} \sech^{4}\frac{z}{2}\tanh^{2}\frac{z}{2}dz\\
=-\frac{36}{a^2}\left(\lambda_M'I_1(\eta-p)+\lambda_MI_1(\eta)
	+\lambda_E'I_2(\eta-p)+\lambda_EI_2(\eta)\right)-\frac{6\gamma}{5a^2},
\end{gathered}
$$
$$
\begin{gathered}
M_{12}=-\int_{-\infty}^{\infty} \frac{3}{2a}  \sech^{2}\frac{z}{2}\tanh\frac{z}{2}\sin\Omega zdz
=-\frac{6\pi\Omega^2}{a\sinh\Omega\pi},
\end{gathered}
$$
$$
\begin{gathered}
M_{21} (\eta)=-\int_{-\infty}^{\infty}  \frac{9\lambda_{M}}{8a^2} \sech^{2}\frac{z}{2}\tanh\frac{z}{2}\sech^{4}\frac{z-\eta}{2}dz\\
+\int_{-\infty}^{\infty} \frac{9\lambda_{M}}{4a^{2}} \sech^{2}\frac{z}{2}\tanh\frac{z}{2}\sech^{2} \frac{z-\eta}{2}\tanh^{2}\frac{z-\eta}{2}dz\\
-\int_{-\infty}^{\infty} \frac{9\lambda_{M}'}{8a^2} \sech^{2}\frac{z}{2}\tanh\frac{z}{2}\sech^{4}\frac{z-\eta+p}{2}dz\\
+\int_{-\infty}^{\infty} \frac{9\lambda_{M}'}{4a^{2}} \sech^{2}\frac{z}{2}\tanh\frac{z}{2}\sech^{2} \frac{z-\eta+p}{2}\tanh^{2}\frac{z-\eta+p}{2}dz\\
-\int_{-\infty}^{\infty} \frac{3\lambda_{E}}{2a^{2}}\sech^{2}\frac{z}{2}\tanh\frac{z}{2}dz
+\int_{-\infty}^{\infty} \frac{9\lambda_{E}}{4a^{2}} \sech^{2}\frac{z}{2}\tanh\frac{z}{2}\sech^{2}\frac{z-\eta}{2}dz\\
-\int_{-\infty}^{\infty} \frac{3\lambda_{E}'}{2a^{2}} \sech^{2}\frac{z}{2}\tanh\frac{z}{2}dz
+\int_{-\infty}^{\infty} \frac{9\lambda_{E}'}{4a^{2}} \sech^{2}\frac{z}{2}\tanh\frac{z}{2}\sech^{2}\frac{z-\eta+p}{2}dz\\
+\int_{-\infty}^{\infty} \frac{9\gamma}{4a^{2}}\sech^{4}\frac{z}{2}\tanh^{2}\frac{z}{2}dz\\
=-\frac{36}{a^2}\left(\lambda_MI_1(-\eta)+\lambda_M'I_1(-\eta+p)
+\lambda_EI_2(-\eta)+\lambda_E'I_2(-\eta+p)\right)+\frac{6\gamma}{5a^2},
\end{gathered}
$$
$$
\begin{gathered}
M_{22}=-\int_{-\infty}^{\infty} \frac{3}{2a}\sech^{2}\frac{z}{2}\tanh\frac{z}{2}\sin\Omega zdz=-\frac{6\pi\Omega^2}{a\sinh\Omega\pi}
\end{gathered}
$$
where
$$
\begin{gathered}
I_1(z)=\frac{(5+50\eu^z-50\eu^{3z}-5\eu^{4z}+z(1+26\eu^z+66\eu^{2z}
+26\eu^{3z}+\eu^{4z}))\eu^z}{(1-\eu^z)^6},\\
I_2(z)=\frac{(3-3\eu^{2z}+z(1+4\eu^z+\eu^{2z}))\eu^z}{(1-\eu^z)^4}.
\end{gathered}
$$

\section*{Appendix B}
The integral terms in \eqref{ins2} are given by
$$
\begin{gathered}
N_{11}(\eta)=\int_{-\infty}^{\infty}  \frac{\lambda_{M}'}{\sqrt{2}}\sech^{2}\frac{z}{\sqrt{2}}\sech^{2} \frac{z-\eta-p}{\sqrt{2}}\tanh\frac{z-\eta-p}{\sqrt{2}}dz\\
+\int_{-\infty}^{\infty} \frac{\lambda_{M}}{\sqrt{2}} \sech^{2}\frac{z}{\sqrt{2}}\sech^{2} \frac{z-\eta}{\sqrt{2}}\tanh\frac{z-\eta}{\sqrt{2}}dz\\
-\int_{-\infty}^{\infty} \frac{\lambda_{E}}{\sqrt{2}}\sech^{2}\frac{z}{\sqrt{2}}\tanh\frac{z-\eta}{\sqrt{2}}dz
-\int_{-\infty}^{\infty} \frac{\gamma}{2}\sech^{4}\frac{z}{\sqrt{2}}dz\\
-\int_{-\infty}^{\infty} \frac{\lambda_{E}'}{\sqrt{2}}\sech^{2}\frac{z}{\sqrt{2}}\tanh\frac{z-\eta-p}{\sqrt{2}}dz\\
=8\left(\lambda_M'J_1(-\eta-p)+\lambda_MJ_1(-\eta)\right)
	+2\left(\lambda_E'J_2(-\eta-p)+\lambda_EJ_2(-\eta)\right)
	-\frac{2\sqrt{2}\gamma}{3},
\end{gathered}
$$
$$
N_{12}=\int_{-\infty}^{\infty} \frac{1}{\sqrt{2}} \sech^{2}\frac{z}{\sqrt{2}}\cos\Omega zdz
=\sqrt{2}\pi\Omega\csch\frac{\Omega\pi}{\sqrt{2}},
$$
$$
\begin{gathered}
N_{21} (\eta)=\int_{-\infty}^{\infty} \frac{\lambda_{M}}{\sqrt{2}} \sech^{2}\frac{z}{\sqrt{2}}\sech^{2} \frac{z+\eta}{\sqrt{2}}\tanh\frac{z+\eta}{\sqrt{2}}dz\\
+\int_{-\infty}^{\infty}\frac{\lambda_{M}'}{\sqrt{2}}\sech^{2}\frac{z}{\sqrt{2}}\sech^{2} \frac{z+\eta+p}{\sqrt{2}}\tanh\frac{z+\eta+p}{\sqrt{2}}dz\\
-\int_{-\infty}^{\infty} \frac{\lambda_{E}}{\sqrt{2}}\sech^{2}\frac{z}{\sqrt{2}}\tanh\frac{z+\eta}{\sqrt{2}}dz
-\int_{-\infty}^{\infty} \frac{\lambda_{E}'}{\sqrt{2}}\sech^{2}\frac{z}{\sqrt{2}}\tanh\frac{z+\eta+p}{\sqrt{2}}dz\\
+\int_{-\infty}^{\infty} \frac{\gamma}{2}\sech^{4}\frac{z}{\sqrt{2}}dz\\
=8\left(\lambda_MJ_1(\eta)+\lambda_M'J_1(\eta+p)\right)
	+2\left(\lambda_EJ_2(\eta)+\lambda_E'J_2(\eta+p)\right)
	+\frac{2\sqrt{2}\gamma}{3},
\end{gathered}
$$
$$
N_{22}=\int_{-\infty}^{\infty} \frac{1}{\sqrt{2}} \sech^{2}\frac{z}{\sqrt{2}}\cos\Omega zdz
=\sqrt{2}\pi\Omega\csch\frac{\Omega\pi}{\sqrt{2}}$$
where
$$
\begin{gathered}
J_1(z)=\frac{\left(3-3\eu^{2\sqrt{2}z}
	+\sqrt{2}z\left(1+4\eu^{\sqrt{2}z}+\eu^{2\sqrt{2}z}\right)\right)
	\eu^{\sqrt{2}z}}
	{(1-\eu^{\sqrt{2}z})^4},\\
J_2(z)=\frac{1-\eu^{2\sqrt{2}z}+2\sqrt{2}z\eu^{\sqrt{2}z}}
	{(1-\eu^{\sqrt{2}z})^2},
\end{gathered}
$$
and 
$$
\begin{gathered}
N_{11\pm}=\pm\int_{-\infty}^{\infty} \frac{\lambda_{E}+\lambda_{E}'}{\sqrt{2}}\sech^{2}\frac{z}{\sqrt{2}}dz-\int_{-\infty}^{\infty} \frac{\gamma}{2}\sech^{4}\frac{z}{\sqrt{2}}dz\\
=\pm 2(\lambda_{E}+\lambda_{E}')-\frac{2\sqrt{2}\gamma}{3},\\
N_{21\pm}=\mp\int_{-\infty}^{\infty} \frac{\lambda_{E}+\lambda_{E}'}{\sqrt{2}}\sech^{2}\frac{z}{\sqrt{2}}dz+\int_{-\infty}^{\infty} \frac{\gamma}{2}\sech^{4}\frac{z}{\sqrt{2}}dz\\
=\mp2(\lambda_{E}+\lambda_{E}')+\frac{2\sqrt{2}\gamma}{3}.
\end{gathered}
$$

%\bibliographystyle{elsarticle-num}
%\bibliographystyle{plainnat}
%\bibliography{DPSurvey}
%\bibliographystyle{authordate1}

\end{document}